\documentclass[11pt]{article}
\usepackage[margin=0.75in]{geometry}
\pdfoutput=1
\usepackage[utf8]{inputenc}
\usepackage[english]{babel}
\usepackage[T1]{fontenc}
\usepackage{amsmath,amssymb}
\usepackage{graphicx}           
\usepackage{relsize}
\usepackage{tcolorbox}
\usepackage{amsthm}
\usepackage{thm-restate}
\usepackage{braket}
\usepackage{csquotes}

\usepackage{soul}
\usepackage{mathtools}

\usepackage[all]{xy}

\usepackage{cite}

\theoremstyle{definition}
\usepackage{float}
\makeatletter
\newlength\min@xx

\makeatother

\usepackage[figurename=Fig.,font=footnotesize]{caption}

\usepackage{hyperref}
\hypersetup{
    colorlinks=true,
    linkcolor=blue,
    citecolor=magenta,
    filecolor=magenta,      
    urlcolor=blue,
    pdftitle={},
    pdfpagemode=FullScreen,
    }

\usepackage{lipsum}

\newtheorem{theorem}{Theorem}
\newtheorem{lemma}{Lemma}
\newtheorem{corollary}{Corollary}
\newtheorem{definition}{Definition}
\newtheorem*{remark}{Remark}
\newtheorem{proposition}{Proposition}
\newtheorem{example}{Example}

\def\ket#1{| #1 \rangle}

\newcommand{\bbC}{\mathbb{C}}
\newcommand{\bbZ}{\mathbb{Z}}
\newcommand{\bbT}{\mathbb{T}}

\newcommand{\calB}{\mathcal{B}}
\newcommand{\calC}{\mathcal{C}}

\newcommand{\calG}{\mathcal{G}}
\newcommand{\calH}{\mathcal{H}}
\newcommand{\calK}{\mathcal{K}}
\newcommand{\calL}{\mathcal{L}}
\newcommand{\calN}{\mathcal{N}}
\newcommand{\calP}{\mathcal{P}}

\newcommand{\calU}{\mathcal{U}}

\newcommand{\Ind}{\operatorname{Ind}}
\newcommand{\Res}{\operatorname{Res}}
\newcommand{\Hom}{\operatorname{Hom}}
\newcommand{\Span}{\operatorname{span}}
\newcommand{\Tr}{\operatorname{Tr}}
\newcommand{\id}{\operatorname{id}}

\DeclareMathOperator{\wt}{wt}

\title{Hybrid Clifford Codes via Operator Algebra Quantum Error Correction and Projective Representation Theory}
\author{Jonas Eidesen$^1$, David W. Kribs$^{2}$\footnote{Corresponding author: \href{mailto:dkribs@uoguelph.ca}{dkribs@uoguelph.ca}} and Andrew Nemec$^3$}
\date{\small
$^1$Department of Mathematics, University of Oslo, 
P.O.Box 1053 Blindern,
0316 Oslo,
Norway \\
$^2$Department of Mathematics \& Statistics, University of Guelph, Guelph, ON N1G 2W1, Canada \\ 
$^3$Department of Computer Science, University of Texas at Dallas, 890 Franklyn Jenifer Drive, Richardson, TX 75080 USA \\
\normalsize \today
}

\begin{document}

\maketitle

\begin{abstract}
Clifford codes are a natural generalization of quantum stabilizer codes based primarily on representation theory. This class of codes has previously been extended to the setting of quantum subsystem codes. We formulate a two-fold generalization of Clifford codes, for both the hybrid classical and quantum information and projective representation theory settings. This leads to new classes of hybrid subspace and subsystem Clifford codes. We extend the fundamental representation theoretic quantum error correction theorem to include these codes, based on the operator algebra quantum error correction framework. We also discuss several examples throughout the presentation, of both stabilizer and non-stabilizer type. 
\end{abstract}

\section{Introduction}

Quantum error correction has grown over the past three decades to become central to both the ongoing and continued development of quantum technologies, and a core area in quantum information theory \cite{nielsen_chuang_2010,lidar2013quantum}. Recent high-profile advances in the quantum computing industry have been based on several important factors, with a strong foundation of theoretical quantum error correction being one of them. Future advances will surely be built on additional theoretical advances, and so there are compelling reasons to broaden and deepen the theory further, including in quantum error correction. 

Historically, the most important class of quantum codes has been the class of stabilizer codes \cite{gottesman1996class,gottesman1997stabilizer}, originally built for error models made up of operators from the Pauli group. An early generalization of these codes extended the formulation of stabilizer codes based on a group representation theory approach, and this led to what are now called Clifford codes \cite{KnillI96,KnillII96,KlappeneckerRotteler02,KlappeneckerRottelerII02}. Subsequently, the original class of Clifford codes were generalized to a more exotic type of quantum code called subsystem codes \cite{klappenecker2008clifford,klappenecker2010clifford}, which have been intensely investigated over the past two decades \cite{bacon2006operator,Poulin2005Stabilizer,Kribs2005Unified}. 
In a different direction, and including more recent work, error correction approaches have been developed for hybrid classical and quantum codes \cite{devetak2005capacity,hsieh2010entanglement,hsieh2010trading,grassl2017codes,li2020error,cao2021higher,nemec2021infinite,nemec2022encoding,majidy2018unification}, which have a strong connection with the theory of operator algebras \cite{kuperberg2003capacity,beny2009quantum,crann2016private}. Of particular relevance to the present work, we mention the operator algebra quantum error correction framework \cite{Beny2007Generalization,Beny2007Quantum}, and the recent extension of the stabilizer formalism to a setting that includes hybrid (subspace and subsystem) codes \cite{DauphinaiKribsVasmer2024,NadkarniAdonsouDauphinaisKribsVasmer2024}. 

Taken together, these earlier and more recent motivations suggest an extension of the class of Clifford codes to hybrid and operator algebra error correction settings would be worthwhile and should be possible. In this paper, we introduce such an extension of Clifford codes, one that embraces hybrid versions of the codes and more directly links them with operator algebra theory. Furthermore, our generalization is two-fold, as we introduce another new element into the study of these codes: projective representation theory, which has a natural fit in quantum information theory, but has only recently been rigorously applied in the subject \cite{Eidesen2025}. 
Following our formulation of the code class and establishing basic results, we establish an extension of the fundamental quantum error correction theorem to the setting. 
We present several examples, of both stabilizer and non-stabilizer type, which we carry through the presentation, and we include a related example that we show is not a Clifford code. We also discuss how code recovery works and we introduce a distance notion for the codes. 

This paper is organized as follows. The next section includes a brief review of notions required for our work, including an extended review of key notions from projective representation theory. We then introduce a collection of examples in Section~3 that will be revisited throughout the rest of the paper. In Section~4, we formulate and define hybrid (subspace and subsystem) Clifford codes, based on the projective representation theory approach. In Section~5, we use the operator algebra quantum error correction framework to prove the error correction theorem for the codes, deriving relavent details on the code space projections as part of the analysis. We finish in Section~6 with some concluding remarks.

\section{Preliminaries}\label{prelims}

In quantum coding theory \cite{knill1997theory,nielsen_chuang_2010,lidar2013quantum}, one is typically  interested in finite-dimensional Hilbert spaces $\mathcal H$ made up of multiple qubit (i.e., $\mathbb C^2$) or qudit (i.e., $\mathbb C^d$ with $d > 2$) subsystems, and so $\mathcal H = (\mathbb C^d)^{\otimes n}$. This is called $n$-qubit or $n$-qudit Hilbert space, and we denote the standard (computational) orthonormal basis by $\{ \ket{i_1 \cdots i_n} \, : \, 0 \leq i_j \leq d-1\}$, with the convention $\ket{i_1 \cdots i_n} := \ket{i_1} \otimes \cdots \otimes \ket{i_n}$. Quantum codes are usually identified with subspaces of a quantum system Hilbert space, with generalized notions defined by a subsystem of a subspace (for subsystem codes), or a collection of orthogonal subspaces (for hybrid codes). 
We will introduce the specific code families and the relevant theoretical aspects of quantum error correction for our investigation in the sections that follow; here we simply include some introductory notes. The original formulation of Clifford codes \cite{KnillI96,KnillII96,KlappeneckerRotteler02,KlappeneckerRottelerII02} in quantum error correction came from a group representation theory motivated generalization of stabilizer codes \cite{gottesman1996class,gottesman1997stabilizer}, which are historically the most important class of quantum codes. Subsequently, they were generalized to subsystem code constructions \cite{klappenecker2008clifford,klappenecker2010clifford}, which is an important extension of regular (subspace) quantum codes that have grown in importance over the past two decades \cite{bacon2006operator,Poulin2005Stabilizer,Kribs2005Unified}. 

We are thus motivated to consider finite groups $G$ represented on finite-dimensional Hilbert space $\mathcal H$. The set of (linear) operators on $\mathcal H$ will be denoted by $\mathcal B(\mathcal H)$, and the group of unitary operators by $\mathcal U(\mathcal H)$. The set of $N\times N$ complex matrices is denoted $M_N$, with identity operator $I_N$. If $\dim \mathcal H = N$, then $\mathcal B(\mathcal H)$ can be identified with $M_N$ via matrix representations in a fixed orthonormal basis for $\mathcal H$. 
Given a subgroup of unitary operators $\mathcal L$ inside $\mathcal U(\mathcal H)$, let  $\mathcal A$ be the subalgebra of $\mathcal B(\mathcal H)$ generated by $\mathcal L$, which is equal to the set of complex polynomials in the elements of $\mathcal L$. The algebra $\mathcal A$ is a finite-dimensional von Neumann algebra or (unital) C$^*$-algebra \cite{davidson1996c,paulsen2002completely}. From the structure theory for such algebras, $\mathcal A$ is unitarily equivalent to a direct sum of the form 
\[
\mathcal A \simeq \bigoplus_i (M_{m_i} \otimes I_{n_i}), 
\]
for some positive integers $m_i, n_i$  with $\sum_i m_i n_i = \dim \mathcal H$. This unitary equivalence defines a decomposition of the Hilbert space $\mathcal H$ as an orthogonal direct sum of subspaces, each with its own corresponding tensor decomposition, $\mathcal H = \oplus_i (A_i \otimes B_i)$, in which the algebra decomposes as $\mathcal A = \oplus_i (\calB(A_i) \otimes I_{B_i})$. In addition, the commutant $\mathcal A^\prime$ of the algebra, the algebra of all operators that commute with all elements of $\mathcal A$, is unitarily equivalent to 
\[
    \mathcal A^\prime \simeq \bigoplus_i (I_{m_i} \otimes M_{n_i}), 
\] 
which is also determined by the structure of the Hilbert space decomposition as ${\mathcal A}^\prime = \oplus_i (I_{A_i} \otimes \calB(B_i))$. 

Operator algebras have long been recognized as a vehicle for encoding hybrid forms of classical and quantum information \cite{kuperberg2003capacity}.
The basic idea is that quantum codes can be encoded into each of the non-trivial matrix algebras in the above sum, and the orthogonality of the algebras in the summand can be used to also include classical information. 
A generalized approach called `operator algebra quantum error correction' \cite{Beny2007Generalization,Beny2007Quantum}, was built on this idea and provided a framework for considering the simultaneous encoding of classical and quantum information \cite{devetak2005capacity,hsieh2010entanglement,hsieh2010trading,grassl2017codes,li2020error,cao2021higher,nemec2021infinite,majidy2018unification}, and for building error correction protocols that also extend to the infinite-dimensional setting \cite{beny2009quantum,crann2016private}.

\subsection{Projective Representation Theory}

As projective representation theory has only recently been considered in a quantum error correction context \cite{Eidesen2025}, for the rest of this section we will review key elements from the theory (for further details see \cite{CeccheriniSilbersteinTullio22}). We have also included a remark below on the use of this theory in quantum error correction. 

\begin{definition}\label{projrepdefn}
    Let $G$ be a finite group. A {\it projective representation} of $G$ on a finite-dimensional Hilbert space $\calH$ is a function $\pi \colon G \to \calU(\calH)$, where $\calU(\calH)$ is the unitary group on $\mathcal H$, such that the composition $q \circ \pi \colon G \to P\calU(\calH)$ is a group homomorphism, where $P\calU(\calH)$ is the quotient of $\calU(\calH)$ by scalar multiples of the identity and $q \colon \calU(\calH) \to P\calU(\calH)$ is the quotient map. We say that $\pi$ is \emph{projectively faithful} if the composition $q \circ \pi$ is injective. 
\end{definition} 

\begin{remark}\label{rem:reason for projective}
    The reason it is desirable to work with \emph{projective} representations in quantum information is as follows. To implement measurements in quantum mechanics one must compute various inner products  $\langle\varphi\vert\psi\rangle$. This quantity remains unchanged if we change $\ket{\psi}$ and $\ket{\varphi}$ by the same global phase, i.e. if we act by a unimodular scalar multiple of the identity on $\calH$. Thus, all the unimodular scalar multiples of the identity are irrelevant information for us to remember, in other words, the physics we want to describe is completely captured by the group homomorphism $q \circ \pi \colon G \to P\calU(\calH)$.

    A natural question might be, why do we not simply work with these group homomorphisms to $P\calU(\calH)$ directly? The answer is that this does not constitute a good representation theory, certain theoretical results simply do not hold for such group homomorphisms, but they do hold for projective representations. An example of such a phenomenon is highlighted in \cite[Section 3]{Eidesen2025}, where it is shown that it is possible to have inequivalent projective representations $\pi_1$ and $\pi_2$ such that $q \circ \pi_1 = q \circ \pi_2$.
\end{remark}

That the composition $q \circ \pi$ in Definition~\ref{projrepdefn} is a group homomorphism is equivalent to the existence of a function $\sigma \colon G \times G \to \bbT$, where $\bbT$ is the unit circle in the complex plane, such that
\begin{equation}\label{projectiverelation}
    \pi(x)\pi(y) = \sigma(x,y)\pi(xy), \text{ for all } x,y \in G.
\end{equation}
We note that this differs slightly from the definition in \cite{CeccheriniSilbersteinTullio22} where the $\sigma$-function is conjugated. 

One can readily verify from Eq.~(\ref{projectiverelation}) that $\sigma$ satisfies the following cocycle equation for all $x,y,z \in G$:
\begin{equation}\label{cocycle}
    \sigma(x,y)\sigma(xy,z) = \sigma(x,yz)\sigma(y,z),
\end{equation}
the function $\sigma$ is therefore aptly named a \emph{cocycle}. We will often use the notation $\sigma_\pi$ to denote the cocycle associated to the projective representation $\pi$. Write the identity element of $G$ as $e$ and note that $\sigma(x,e)$ and $\sigma(e,x)$ are equal and the same constant for all $x\in G$ by \eqref{cocycle}. So, without loss of generality we will assume $\sigma(x,e) = \sigma(e,x) = 1$ for all $x\in G$. This in turn implies from the cocycle equation that $\sigma(x,x^{-1}) = \sigma(x^{-1}, x)$ for all $x\in G$. The assumption also implies that $\pi(e)$ is a projection (i.e., not just a unimodular scalar multiple of a projection, allowing us to avoid having to deal with that degeneracy), which is the identity operator when $\pi$ is irreducible. Further, for all $x\in G$, we have 
\begin{equation}\label{inverse}
\pi(x^{-1}) = \sigma(x,x^{-1}) \pi(x)^{-1} = \sigma(x^{-1}, x) \pi(x)^{-1}.
\end{equation}

We will need a few more notions from projective representation theory for our analysis, but we will keep the presentation minimal and encourage the interested reader to follow up with \cite{CeccheriniSilbersteinTullio22,Eidesen2025} for more details. 

If $\pi_i \colon G \to \calU(\calH_i)$, $i = 1,2$, are two $\sigma$-projective representations of $G$, we define an \emph{intertwiner} from $\pi_1$ to $\pi_2$ to be a linear map $T \colon \calH_1 \to \calH_2$ such that 
\[
    T\pi_1(x) = \pi_2(x)T
\]
for all $x \in G$. We will often use the notation $T \colon \pi_1 \to \pi_2$ to emphasize that $T$ is an intertwiner and not only a linear map. The collection of all intertwiners between $\pi_1$ and $\pi_2$ is denoted by $\Hom_G(\pi_1,\pi_2)$, and it is clear that this has the natural structure of a complex vector space. We say that $\pi_1$ is isomorphic to $\pi_2$ if there exists an intertwiner $T \in \Hom_G(\pi_1,\pi_2)$ that is also a bijection. In the case where $\pi_1$ is isomorphic to $\pi_2$ we will sometimes write $\pi_1 \simeq \pi_2$.

The collection of $\sigma$-projective representations of $G$ along with intertwiners form a category. This category can be identified with modules over $\bbC[G]^\sigma$, where $\bbC[G]^\sigma$ denotes the \emph{$\sigma$-twisted group algebra of $G$}. As a complex vector space it is defined to be the collection of functions $F \colon G \to \bbC$. It becomes an algebra when endowed with the $\sigma$-twisted convolution product defined by
\begin{equation*}
    (F_1 *_\sigma F_2)(x) = \sum_{y \in G}\sigma(y,y{-1}x)F_1(y)F_2(y^{-1}x).
\end{equation*}
We can endow $\bbC[G]^\sigma$ with an inner product defined by
\begin{equation}\label{eqn:inner product}
    \langle F_1, F_2 \rangle = \frac{1}{|G|} \sum_{x \in G}F_1(x) \overline{F_2(x)}, \text{ for } F_1,F_2 \in \bbC[G]^\sigma.
\end{equation}
With this $\bbC[G]^\sigma$ may also be regarded as a Hilbert space, which is often denoted by $\ell^2(G)$. Note that $\bbC[G]^\sigma$ has a canonical basis given by the elements of $G$, i.e. we can think of $g \in G$ as a function to the complex numbers defined by
\[
    g(x) =
    \begin{cases}
        1, & x = g, \\
        0, & x \neq g.
    \end{cases}
\]

If $\pi$ is a projective representation of a group $G$ on a Hilbert space $\calH$, we say that a subspace $\calC$ of $\calH$ is \emph{$\pi$-invariant} if $\pi(x)\calC \subset \calC$ for all $x \in G$. Then $\pi$ is said to be \emph{irreducible} if the only $\pi$-invariant subspaces of $\calH$ are $\calH$ and $\{0\}$; in particular, $\pi$ is irreducible if and only if $\Span\{\pi(x) \colon x \in G \} = \calB(\calH)$. If $\calC$ is a $\pi$-invariant we will use the notation $\pi|_\calC$ to denote the projective representation attained by restricting $\pi(x)$ to $\calC$ for every $x \in G$. If $P_\calC$ denotes the orthogonal projection onto $\calC$ then
\[
    \pi|_\calC(x) = P_\calC\pi(x)P_\calC.
\]

Projective representations satisfy analogs of Maschke's Theorem and Schur's Lemma, cf. \cite[Corollary 7.15]{CeccheriniSilbersteinTullio22}. Maschke's Theorem says that any projective representation $\pi$ of $G$ is isomorphic to a direct sum of irreducible projective representations. Schur's Lemma states that if $\pi_1$ and $\pi_2$ are irreducible projective representations of $G$, then
\begin{equation*}
    \dim \Hom_G(\pi_1,\pi_2) =
    \begin{cases}
        1 & \text{if } \pi_1 \simeq \pi_2, \\
        0 & \text{if } \pi_1 \not\simeq \pi_2.
    \end{cases}
\end{equation*}

Projective representations also enjoy a character theory which can be used to make powerful arguments. We will need the following result, see \cite[Proposition 2.2]{Cheng15}. If $\pi_1$ and $\pi_2$ are $\sigma$-projective representations of $G$ (not necessarily irreducible), then
\begin{equation}\label{eqn:inner porduct of characters}
    \langle \chi_{\pi_1}, \chi_{\pi_2} \rangle = \dim \Hom_G(\pi_1,\pi_2).
\end{equation}
Here the inner product is the one defined in equation \eqref{eqn:inner product}, and for $i = 1,2$, $\chi_{\pi_i} \in \bbC[G]^\sigma$ are the functions defined by
\begin{equation*}
    \chi_{\pi_i}(x) = \Tr(\pi_i(x)).
\end{equation*}

\subsubsection{Induced projective representations}

Suppose that $H$ is a subgroup of $G$ and $\pi$ is a $\sigma$-projective representation of $G$ on $\calH$. We define a projective representation $\Res^G_H \pi$ of $H$ on $\calH$ by letting
\begin{equation*}
    \Res^G_H\pi(x) = \pi(x), \text{ for } x \in H.
\end{equation*}
We have that $\Res^G_H \pi$ is $\Res^G_H\sigma$-projective, where $\Res^G_H\sigma$ is defined by restricting $\sigma$ to $H \times H$. 

Induction is a procedure that is adjoint to restriction. Concretely: if $\theta$ is a $\Res^G_H\sigma$-projective representation of $H$ on $\calK$, we define a $\sigma$-projective representation $\Ind_H^G\theta$ of $G$ on $\Ind_H^G(\calK)$, where $\Ind_H^G(\calK)$ is the balanced tensor product:
\begin{equation*}
    \Ind_H^G(\calK) := \bbC[G]^\sigma \otimes_H \calK.
\end{equation*}
The subscript $H$ indicates that for any $F \in \bbC[G]^\sigma$, $\ket{\psi} \in \calK$, and $x \in H$ we have that
\begin{equation*}
    F \otimes (\theta(x) \ket{\psi}) = (F *_\sigma x) \otimes \ket{\psi}.
\end{equation*}
If $\{\ket{i}\}_{i = 1}^n$ is a basis for $\calK$, this property ensures that $\{ r \otimes \ket{i} : r \in G/H, \ i \in \{1,\dots,n\} \}$ is a basis for $\Ind_H^G(\calK)$.
Then for any $x \in G$ we define $\Ind_H^G\theta(x)$ by linearly extending the following identity:
\begin{equation*}
    (\Ind_H^G\theta(x))(F \otimes \ket{\psi}) = (x *_\sigma F) \otimes \ket{\psi}, \text{ for } F \in \bbC[G]^\sigma, \ket{\psi} \in \calK.
\end{equation*}
Since the $\sigma$-twisted convolution product in $\bbC[G]^\sigma$ is associative, $\Ind^G_H\theta \colon G \to \calU(\Ind^G_H(\calK))$ indeed defines a $\sigma$-projective representation of $G$.

Induced representations satisfy Frobenius reciprocity: there is a natural isomorphism
\begin{equation*}
    \Hom_G(\Ind_H^G\theta ,\pi) \simeq \Hom_H(\theta, \Res^G_H\pi)
\end{equation*}
for any $\sigma$-projective representation $\pi$ of $G$ and any $\Res^G_H\sigma$-projective representation $\theta$ of $H$, cf. \cite[Theorem 8.15]{CeccheriniSilbersteinTullio22}.

\section{Examples of Projective Error Models}\label{section:examples}

In this section, we will introduce a number of examples and note the different projective representations and cocycles. We will revisit these examples at certain places below to help elucidate different concepts and results. 
We first recall the following definition from \cite{Eidesen2025}.

\begin{definition}\label{def:projective error model}
    Let $\calH$ be a finite dimensional Hilbert space. A \emph{projective error model on $\calH$} is given by a pair $(G,\pi)$ where $G$ is a finite group and $\pi$ is a projectively faithful, irreducible, projective representation of $G$ on $\calH$.
\end{definition}

\begin{remark}\label{rem:relation to nice error frames/basis}
    Note that the data of a projective error model is essentially the same as that of a \emph{nice error frame} introduced in \cite{ChienWaldron17}. In the special case where $|G| = (\dim\calH)^2$, we retrieve what is essentially the same data as a \emph{nice error basis}, cf.~\cite{KnillI96, KnillII96}.
\end{remark}

Every faithful irreducible unitary representation of a finite group defines a faithful irreducible projective unitary representation, after factoring out the scalar multiples of the identity. This helps motivate the most well-known example of an error model. 

\begin{example}\label{example:Pauli projective error model}
    Consider the $n$-qubit Pauli group $G= P_n$, for integers $n\geq 1$, which is the unitary subgroup on $\mathcal H = (\mathbb C^2)^{\otimes n}$ generated by the scalars $\langle iI\rangle$ and $n$-tensors of the operators, 
        \begin{equation*}
            X =
            \begin{bmatrix}
                0 & 1 \\
                1 & 0
            \end{bmatrix}
            \text{ and }
            Z =
            \begin{bmatrix}
                1 & 0 \\
                0 & -1
            \end{bmatrix} .
        \end{equation*}
        
    We will refer to $P_n$ here as the {\it Pauli error model}. Then a projective error model associated with this error model is given by $(G,\pi)$ where $G = (\mathbb Z_2 \times \mathbb Z_2)^n$ and $\pi$ is the $n$-fold product representation of the (unique) projectively faithful irreducible projective representation of $\mathbb Z_2 \times \mathbb Z_2$. To be concrete, we define 
    $\pi \colon (\mathbb Z_2 \times \mathbb Z_2)^n \to \calU((\mathbb C^2)^{\otimes n})$ by 
    \begin{equation}\label{canonical} 
        \pi (a_1,b_1,\ldots , a_n, b_n) = X^{a_1} Z^{b_1} \otimes \ldots \otimes X^{a_n} Z^{b_n}. 
    \end{equation} 
    In this case, the associated cocycle is the power of $-1$ given by the appropriately defined scalar product of elements in the group, reflecting the anti-commutation relation $XZ = -ZX$. Concretely,
    \[
        \sigma_{\pi}(a_1,b_1,\ldots,a_n,b_n,c_1,d_1,\ldots,c_n,d_n) = \prod_{i = 1}^{n} (-1)^{b_ic_i}.
    \]
\end{example}

The following example is motivated by recent work in quantum error correction. Certain error models and codes defined by the operators lead to what is called the $XP$ (and $XS$ when $d=4$) stabilizer formalism \cite{ni2015non,Webster2022}. This error model was also investigated in the projective representation theory context in \cite{Eidesen2025}.  

\begin{example}\label{example:XP-projective error model}
    Let $d \geq 2$ be an integer, $\zeta_d \in \bbC$ be a primitive $d$-th root of unity, and consider the dihedral group of $2d$ elements with the following presentation:
    \begin{equation*}
        D_d 
        = \langle a,b \, \vert \, a^d = b^2 = 1, \, bab = a^{d-1} \rangle 
        = \{ 1, a, a^2, \dots, a^{d-1}, b, ba, \dots, ba^{d-1} \}.
    \end{equation*}
    Then the function $\pi \colon D_d \to \calU(\bbC^2)$ defined by
    \begin{equation}\label{xpexample}
        \pi(b^ka^l) = X^k P^l, \text{ for } k = 0,1 \text{ and } l = 0,1,\dots,d-1,
    \end{equation}
    where 
    \begin{equation*}
        X =
        \begin{bmatrix}
            0 & 1 \\
            1 & 0
        \end{bmatrix},
        \text{ and }
        P =
        \begin{bmatrix}
            1 & 0 \\
            0 & \zeta_d
        \end{bmatrix},
    \end{equation*}
    is a projectively faithful irreducible projective representation of $D_d$ on $\bbC^2$.

In this case, the cocycle associated with the projective representation is given by $\sigma: D_d \times D_d \to \mathbb T$, where 
\[
\sigma ( b^{k_1} a^{l_1} , b^{k_2} a^{l_2}  )  = \zeta_d^{k_2 l_1}. 
\]
    Additionally, as in the Pauli group case above, we can consider $n$-qubit extension of the representation $\pi^{\otimes n} \colon (D_d)^{\otimes n} \to \calU((\bbC^2)^{\otimes n})$, and the corresponding projective error models. 
\end{example}

The following example is found in \cite[Proposition 8.1]{Eidesen2025}, which we also present here in an alternate form. Note that inserting $d = 2$ in the following example retrieves a well-known example of Klappenecker and Rötteler, cf. \cite[Section 10.9]{KlappeneckerRotteler02}.

\begin{example}\label{cdexample}
    Let $d \geq 2$ be an integer and consider the group $G := C_2 \times D_{2d}$ with the following presentation:
    \begin{equation*}
        G = \langle a, b, c \, \vert \, a^{2d} = b^2 = c^2 = [a,c] = [b,c] = 1, \, bab = a^{-1} \rangle.
    \end{equation*}
    Note that $c$ generates $C_2$, while $a$ and $b$ generate $D_{2d}$.
    
    Define the following matricies:
    \begin{equation*}
        X =
        \begin{bmatrix}
            0 & 1 \\
            1 & 0
        \end{bmatrix}, \quad
        Z =
        \begin{bmatrix}
            1 & 0 \\
            0 & -1
        \end{bmatrix}, \quad
        P =
        \begin{bmatrix}
            1 & 0 \\
            0 & \zeta_{2d}
        \end{bmatrix}.
    \end{equation*}
    Here $\zeta_{2d} = e^{2\pi i/2d}$ is a $2d$-th primitive root of unity. We define a projective representation $\pi \colon G \to \calU(\bbC^2 \otimes \bbC^2)$ by
    \begin{equation}\label{cdexampleeqn}
        \pi(c^kb^la^m) = X^kZ^m \otimes X^lP^m.
    \end{equation}
    The pair $(G, \pi)$ is then a projective error model on $\bbC^2 \otimes \bbC^2$, cf. \cite[Proposition 8.1]{Eidesen2025} for a proof. The cocycle, $\sigma_\pi$, associated to this projective representation is the following:
    \begin{equation*}
        \sigma_\pi(c^{k_1}b^{l_1}a^{m_1}, c^{k_2}b^{l_2}a^{m_2}) = (-1)^{m_1k_2}(\zeta_{2n})^{m_1l_2}.
    \end{equation*}
Further, we can consider multi-qubit versions of these representations as in the examples above, except in this case each subsystem is a two-qubit system. 
\end{example}

The following example is also found in \cite[Proposition 8.2]{Eidesen2025}, which we present here in an alternate form.

\begin{example}\label{zlexample}
    Let $n \geq 3$ be an odd number. Let $\bbZ_2$ act on $\bbZ_n \times \bbZ_n$ via inversion, and let $L := (\bbZ_n \times \bbZ_n) \rtimes \bbZ_2$ denote the associated semi-direct product. Let $G := \bbZ_2 \times L$. Set
    \begin{equation*}
        X_n =
        \begin{bmatrix}
            0 & 1 & 0 & \cdots & 0 & 0 \\
            0 & 0 & 1 & \cdots & 0 & 0 \\
            0 & 0 & 0 & \cdots & 0 & 0 \\
            \vdots & \vdots & \vdots & \ddots & \vdots & \vdots \\
            0 & 0 & 0 & \cdots & 0 & 1 \\
            1 & 0 & 0 & \cdots & 0 & 0
        \end{bmatrix}, \,
        Z_n =
        \begin{bmatrix}
            1 & 0 & 0 & \cdots & 0 & 0 \\
            0 & \zeta_n & 0 & \cdots & 0 & 0 \\
            0 & 0 & \zeta_n^2 & \cdots & 0 & 0 \\
            \vdots & \vdots & \vdots & \ddots & \vdots & \vdots \\
            0 & 0 & 0 & \cdots & \zeta_n^{n-2} & 0 \\
            0 & 0 & 0 & \cdots & 0 & \zeta_n^{n-1}
        \end{bmatrix},
        \text{ and }
        C =
        \begin{bmatrix}
            0 & 0 & \cdots & 0 & 1 \\
            0 & 0 & \cdots & 1 & 0 \\
            \vdots & \vdots & \ddots & \vdots & \vdots \\
            0 & 1 & \cdots & 0 & 0 \\
            1 & 0 & \cdots & 0 & 0
        \end{bmatrix},
    \end{equation*}
     and define $\pi \colon G \to \calU(\bbC^2 \otimes \bbC^n)$ by
    \begin{equation}\label{zlexampleeqn}
        \pi(a,b,c,d) = X^a Z^d \otimes (X_n)^b (Z_n)^c C^d.
    \end{equation}
    Here $X$ and $Z$ are the usual Pauli $X$ and Pauli $Z$. Note that $a,d \in \{0,1\}$ while $b,c \in \{0,1,...,n-1\}$. The pair $(G,\pi)$ is then a projective error model on $\bbC^2 \otimes \bbC^n$, cf. \cite[Proposition 8.2]{Eidesen2025} for a proof of this fact. The cocycle, $\sigma_\pi$, associated to this projective representation is given by
    \begin{equation*}
        \sigma_\pi(a_1,b_1,c_1,d_1,a_2,b_2,c_2,d_2) = (-1)^{d_1a_2}(\zeta_n)^{-(-1)^{d_1}c_1b_2 - d_1c_2}.
    \end{equation*}

As above, we can consider multi-qubit versions of these representations, and in this case each subsystem is $2n$-dimensional. 
\end{example}

\section{Hybrid Clifford Subspace and Subsystem Codes}

In this section we formulate our class of Clifford codes, with subclasses defined sequentially as extra features are added in the discussion. 
Reflecting one of our generalization directions, let us first briefly review key elements of the recent construction of hybrid subspace and subsystem codes from \cite{DauphinaiKribsVasmer2024} based on the operator algebra quantum error correction approach. 

In \cite{DauphinaiKribsVasmer2024}, codes are constructed from the following algebraic data of the Pauli group $P_n$. The initial data is the group of stabilizers, $S$, which is an abelian subgroup of $P_n$ that does not contain $-I$. From this one considers a splitting of the normalizer group $\calN(S)$ into a cartesian product $\calN(S) \simeq \calL \times \calG$. The group $\calL$ is the group of (bare) logical operators, and $\calG$ is the so-called gauge group. Then, let $T \subset \calP_n$ be a choice of representatives of the (left) cosets $P_n/\calN(S)$, a so-called coset transversal, and let $T_0 \subseteq T$ be a subset. A code $\calC = \calC(S,\calG,\calL,T_0)$ is then defined to be the following subspace of $\calH = (\bbC^2)^{\otimes n}$,
\begin{equation*}
    \calC(S,\calG,\calL,T_0) = \bigoplus_{U \in T_0} U(A \otimes B),
\end{equation*}
where the subspace $\calC = A \otimes B$ is the stabilizer code associated to $S$, the subsystem structure is attained by the splitting $\calN(S) \simeq \calL \times \calG$, and the hybrid classical feature is given by the transversal subset $T_0$ as discussed further below.

We begin by defining a pair of subgroups of the given group that are key to our code constructions. Note that we use the term `dressed' logical operators here, because when we introduce subsystem structures below it will be evident that, as in other subsystem code settings, these operators can act non-trivially on both logical and gauge qubits. 

\begin{definition}\label{def:stabilizer and logical group of a code}
    Let $\calH$ be a finite dimensional Hilbert space and $(G,\pi)$ be a projective error model on $\calH$. If $\calC \subseteq \calH$ is a subspace (i.e. a quantum code) we define the following subgroups of $G$:
    \begin{align*}
        S_{(G,\pi)}(\calC) & := \{ x \in G : P_\calC \pi(x)P_\calC \in \bbT P_C \}, \\
        L_{(G,\pi)}(\calC) & := \{ x \in G : P_\calC \pi(x) = \pi(x) P_\calC\},
    \end{align*}
    where $P_\calC$ denotes the orthogonal projection onto the subspace $\calC$. The subgroup $L_{(G,\pi)}(\calC)$ is \emph{the group of dressed logical operators of $\calC$}, and the subgroup $S_{(G,\pi)}(\calC)$ is \emph{the group of stabilizers of $\calC$}. It is clear that $S_{(G,\pi)}(\calC)$ is a subgroup of $L_{(G,\pi)}(\calC)$.
\end{definition}

Now we define (non-hybrid, subspace) Clifford codes in the projective representation theory setting. 

\begin{definition}\label{def:Stabilizer code and Clifford code}
    Let $\calH$ be a finite dimensional Hilbert space, $(G,\pi)$ be a projective error model on $\calH$, and $\calC \subseteq \calH$ be a subspace. We say that $\calC$ is a \emph{Clifford subspace code} if there exists a normal subgroup $L$ of $G$ and a projective representation $\gamma \colon L \to \calU(\calC)$ such that $\Ind_L^G\gamma \simeq \pi$. We will use the notation $\calC = \calC (G,\pi, L, \gamma)$ to refer to the Clifford subspace code $\calC$ along with its defining data.

    We say that $\calC$ is a \emph{stabilizer code} if there exists a normal subgroup $S$ of $G$ and a function $\lambda \colon S \to \bbT$ such that
    \begin{equation*}
        \calC = \{ \ket{\psi} \in \calH : \pi(s)\ket{\psi} = \lambda(s)\ket{\psi} \text{ for all } s \in S \}.
    \end{equation*}
    We will sometimes use the notation $\calC = \calC(G,\pi,S,\lambda)$ to refer to the stabilizer code $\calC$ along with its defining data.
\end{definition}

\begin{remark}\label{rem:normality of L}
    The definition of a Clifford subspace code as given above differs slightly from that given in \cite[Definition 5.8]{Eidesen2025}, where it was not required that the subgroup $L$ be normal. We make this assumption here as it simplifies some of the statements and proofs. We furthermore felt justified in doing so as we have seen that $L$ is normal in all examples we have studied so far.

    In the case that $\calC = \calC(G,\pi,S,\lambda)$ is a stabilizer code, then we have $S \subset S_{(G,\pi)}(\calC)$. In this sense, stabilizer codes can be thought of as codes where the group $S_{(G,\pi)}(\calC)$ is prescribed (in the special case where $|G| = (\dim\calH)^2$ we necessarily have $S = S_{(G,\pi)}(\calC)$). The group $L_{(G,\pi)}(\calC)$ are all the elements in $G$ that leave the code space $\calC$ invariant, and it is the group that corresponds to the normalizer $\calN(S)$ when the Pauli error model is considered. 
\end{remark}

\begin{remark}\label{rem:clifford code of abstract rep}
    If $\calC = \calC(G,\pi,L,\gamma)$ is a Clifford subspace code, then we have by Frobenius reciprocity that
    \begin{equation}\label{eqn:clifford code representation}
        \Res^G_L\pi|_\calC = \gamma.
    \end{equation}
    With this we get an equivalent definition of a Clifford subspace code: If $(G,\pi)$ is a projective error model on a finite dimensional Hilbert space $\calH$, then a subspace $\calC \subset \calH$ is a Clifford subspace code if there exists a normal subgroup $L$ of $G$, and a projective representation $\gamma$ of $L$ (not necessarily on $\calC$) such that;
    \begin{equation*}
        \Ind_L^G\gamma \simeq \pi \quad \text{and} \quad \Res^G_L\pi|_\calC \simeq \gamma.
    \end{equation*}
    That this is an equivalent definition is \cite[Proposition 5.10]{Eidesen2025}.
\end{remark}

The following result tells us how Clifford codes and stabilizer codes are related, as well as how $S$ is related to $S_{(G,\pi)}(\calC)$ and how $L$ is related to $L_{(G,\pi)}(\calC)$. See \cite[Proposition 5.12, Theorem 6.3, Proposition 6.7]{Eidesen2025} for the relevant proofs.

\begin{theorem}\label{thm:clifford and stabilizer code facts}
    Let $\calH$ be a finite dimensional Hilbert space and $(G,\pi)$ be a projective error model on $\calH$. The following are true for a subspace code $\calC \subseteq \calH$:
    \begin{enumerate}
        \item If $\calC = \calC(G,\pi,S,\lambda)$ is a stabilizer code it is also a Clifford subspace code in the sense of \cite{Eidesen2025}, i.e. the subgroup $L$ may not be normal.
        \item If $\calC = \calC(G,\pi,L,\gamma)$ is a Clifford subspace code, then
        \[
            L_{(G,\pi)}(\calC)  = L \quad \mathrm{and} \quad 
            S_{(G,\pi)}(\calC)  = \ker(q \circ \gamma);
        \]
        in particular $S_{(G,\pi)}(\calC)$ is also a normal subgroup of $G$.
        \item If $\calC = \calC(G,\pi,S,\lambda)$ is a stabilizer code then $S$ is a subgroup of $S_{(G,\pi)}(\calC)$, and we may extend $\lambda$ to a function $\Tilde{\lambda} \colon S_{(G,\pi)}(\calC) \to \bbT$ such that
        \[
            \calC = \{ \ket{\psi} \in \calH : \pi(s)\ket{\psi} = \Tilde{\lambda}(s)\ket{\psi} \text{ for all } s \in S_{(G,\pi)}(\calC) \}.
        \]
        Note that $S_{(G,\pi)}(\calC)$ might not in general be a normal subgroup of $G$.
    \end{enumerate}
\end{theorem}

A Clifford subspace code can naturally be considered as a summand in a hybrid code. We make this precise in the following result, which shows how the full Hilbert space decomposes as a direct sum determined by the code space and a coset transversal. We shall give a more explicit description of the associated projections in the context of the error correction theorem of the next section. 

\begin{proposition}\label{prop:Hybrid structure of Clifford codes}
    Let $\calH$ be a finite dimensional Hilbert space and $(G,\pi)$ be a projective error model on $\calH$. Suppose that $\calC = \calC(G,\pi,L,\gamma)$ is a Clifford subspace code. and let $T$ denote a coset transversal for the factor group $G/L$. Then $\calH$ decomposes into a direct sum of mutually orthogonal subspaces as follows:
    \begin{equation*}
        \calH = \bigoplus_{t \in T} \pi(t)\calC.
    \end{equation*}
\end{proposition}
\begin{proof}
    By the assumption that $\calC$ is a Clifford subspace code we have that $\pi \simeq \Ind_L^G\gamma$. Hence, if $T$ is a set of transversals for the cosets $G/L$ we have the following isomorphisms:
    \begin{equation*}
        \bigoplus_{t \in T} t \otimes \calC
        \simeq \Ind_L^G\calC
        \simeq \calH.
    \end{equation*}
    The image of $t \otimes \calC$ under these identifications is precisely $\pi(t) \calC$ as a subspace of $\calH$. Thus,
    \begin{equation*}
        \calH = \bigoplus_{t \in T} \pi(t)\calC. \qedhere
    \end{equation*}
\end{proof}

Thus, given a subset of a coset transversal, we can define hybrid versions of our Clifford subspace codes, as was done in \cite{DauphinaiKribsVasmer2024}. Note that what we call `classical logical operators' in the following are, in other settings, sometimes called `translation operators'. Also observe that without loss of generality we can always assume the identity operator is the representative for the trivial coset. 

\begin{definition}\label{def:hybrid Clifford code}
    Let $\calH$ be a finite-dimensional Hilbert space and $(G,\pi)$ be a projective error model on $\calH$. Suppose that $\calC = \calC(G,\pi,L,\gamma)$ is a Clifford code, and $T_0 \subseteq T$ is a subset of a coset transversal for the factor group $G/L$. The subspace
    \[
        \bigoplus_{t\in T_0} \pi(t) \calC
    \]
    is called a \emph{hybrid Clifford subspace code}. We will refer to the subspace $\calC$ as the \emph{base code}, and we will use the notation $\calC = \calC(G,\pi,L,\gamma,T_0)$ to refer to the base code along with its Clifford subspace structure and the hybrid structure coming from the subset $T_0$, which we will refer to as the {\it classical logical operators} for the code.
\end{definition}

As a last step, we will discuss when the base code of a Clifford subspace code $\calC$ naturally decomposes as a tensor product $A \otimes B$, for the purposes of considering subsystem codes in this setting. Note that a subsystem decomposition of the base code of a hybrid code is carried to the other coset subspaces by the unitary actions of $\pi(t)$, for the classical logical operators $t\in T_0$, restricted to $\mathcal C$. Such a decomposition can be guaranteed in the following scenario.

\begin{lemma}\label{lemma:subsystem structure of Clifford code}
    Let $\calH$ be a finite-dimensional Hilbert space and $(G,\pi)$ be a projective error model on $\calH$. Suppose that $\calC = \calC(G,\pi,L,\gamma)$ is a Clifford subspace code. If there exists normal subgroups $\calL,\calG \subset L$ such that $L \simeq \calL \times \calG$, and projective representations $\alpha \colon \calL \to \calU(A)$, $\beta \colon \calG \to \calU(B)$ such that $\gamma \simeq \alpha \otimes \beta$, we have that
    \begin{equation*}
        \calC \simeq A \otimes B.
    \end{equation*}
    Furthermore, if $\alpha'$ is a projective representation of $\calL$ and $\beta'$ is a projective representation of $\calG$ such that $\gamma \simeq \alpha' \otimes \beta'$, then $\alpha \simeq \alpha'$ and $\beta \simeq \beta'$.
\end{lemma}
\begin{proof}
    The first statement that $\calC \simeq A \otimes B$ follows by unwrapping what it means for two projective representations to be isomorphic.

    Now suppose that $\alpha \otimes \beta \simeq \gamma \simeq \alpha' \otimes \beta'$ and let $T \colon \alpha \otimes \beta \to \gamma$ be an isomorphism. Then, for any $l_1, l_2 \in \calL$ we have that
    \[
        (\alpha(l_1) \otimes \beta(e))(\alpha(l_2) \otimes \beta(e)) T
        = T \gamma(l_1) \gamma(l_2)
        = \sigma_\gamma(l_1,l_2) T \gamma(l_1l_2)
        = \sigma_\gamma(l_1,l_2) (\alpha(l_1l_2) \otimes \beta(e)) T.
    \]
    Hence, $\sigma_\alpha = \Res^L_\calL\sigma_\gamma$. By symmetry we hence get that
    \[
        \sigma_\alpha = \sigma_{\alpha'} \quad \text{and} \quad \sigma_\beta = \sigma_{\beta'}.
    \]
    Furthermore, a similar computation shows that for $l_1, l_2 \in \calL$ and $g_1, g_2 \in \calG$, we have that
    \[
        \sigma_\gamma((l_1,g_1),(l_2,g_2)) = \sigma_\alpha(l_1,l_2)\sigma_\beta(g_1,g_2).
    \]
    Adapting \cite[Theorem 10]{Serre77} to the projective setting then shows that $\alpha \simeq \alpha'$ and $\beta \simeq \beta'$.
\end{proof}

\begin{remark} 
    As in other subsystem code contexts, we will refer to $\calL$ as \emph{the bare logical group} and $\calG$ as \emph{the gauge group} for the code. (Recall we refer to $L$ as the dressed logical operators.) In the case where the assumptions in Lemma \ref{lemma:subsystem structure of Clifford code} are satisfied, we say that $\calC$ is a \emph{Clifford subsystem code}, and it is defined by the data $(G,\pi,\calL,\alpha,\calG,\beta)$. Note that the projective representations $\alpha$ and $\beta$ are uniquely determined by the fact that $\calC$ is a Clifford code, this follows by \cite[Corollary 5.4]{Eidesen2025} and Lemma \ref{lemma:subsystem structure of Clifford code} above. Hence, the subsystem structure of $\calC$ is completely determined by the quadruple $(G,\pi,\calL,\calG)$, and we will often only use this quadruple to refer to the associated Clifford subsystem code.
\end{remark} 

Above we gave a top-down definition of a hybrid Clifford subsystem code, where we first start with a Clifford subspace code, and then ask for a splitting of the group of logical operators that satisfy the assumptions of Lemma \ref{lemma:subsystem structure of Clifford code}. We could instead start with a projective error model $(G,\pi)$ on $\calH$ and ask for the existence of two subgroups $\calL \subset G$ and $\calG \subset G$ giving rise to a Clifford subsystem code determined by $(G,\pi,\calL,\calG)$. However, to be able to identify $\calL \times \calG$ as a subgroup of $G$ we need the following result, which follows from standard group theory.

\begin{lemma}\label{lemma:criteria for product subgroup}
    Let $G$ be a finite group and $\calL,\calG \subset G$ be subgroups of $G$. If the elements in $\calL$ commute with all the elements in $\calG$, i.e. $\calL$ and $\calG$ commute, we have that the set
    \begin{equation*}
        \calL \calG = \{ lg : l \in \calL, \text{ and } g \in \calG \}
    \end{equation*}
    is a subgroup of $G$. If we furthermore have that $\calL \cap \calG = \{e\}$, we have that
    \begin{equation*}
        \calL\calG \simeq \calL \times \calG.
    \end{equation*}
\end{lemma}

We can now define subsystem codes in this setting. 

\begin{definition}\label{def:Clifford subsystem code}
    Let $\calH$ be a finite dimensional Hilbert space and $(G,\pi)$ be a projective error model on $\calH$. A subspace $\calC \subseteq \calH$ is a \emph{Clifford subsystem code} if there exists commuting subgroups $\calL,\calG \subset G$ that only intersect at the identity element $e$, and irreducible projective representations $\alpha$ of $\calL$ and $\beta$ of $\calG$ such that $\Ind_{\calL \times \calG}^G(\alpha \otimes \beta) \simeq \pi$, and $\Res^G_{\calL \times \calG}\pi|_\calC \simeq \alpha \otimes \beta$. We will use the notation $A \otimes B = \calC(G,\pi,\calL,\calG)$ to refer to the Clifford subsystem code $\calC$ along with its identification with $A \otimes B$ and its defining data, and we refer to $\mathcal L$ as the bare logical group and $\mathcal G$ as the gauge group.
\end{definition}

\begin{remark}
    Since $\calC = A \otimes B = \calC(G,\pi,\calL,\calG)$ also defines a Clifford subspace code, we have that $S_{(G,\pi)}(\calC)$ is a subgroup of $L = \calL\calG$, since $\calL \cap \calG = \{e\}$ we have that $S_{(G,\pi)}(\calC) \subset \calL$ or $S_{(G,\pi)}(\calC) \subset \calG$. By convention we will always assume that the latter is true, i.e. that $S_{(G,\pi)}(\calC) \subset \calG$.
\end{remark}

We next give our final code definition, which includes all possible elements, along with the associated operator algebra it is identified with. 

\begin{definition}\label{def:hybrid Clifford subsystem code}
    Let $\calH$ be a finite dimensional Hilbert space and $(G,\pi)$ be a projective error model on $\calH$. A subspace $\calC = \calC(G,\pi,\calL,\calG,T_0)$ is a \emph{hybrid Clifford subsystem code} if $\calC(G,\pi,\calL,\calG)$ defines a Clifford subsystem code, and $T_0 \subseteq T$ is a subset of a coset transversal for the factor group $G/\calL\calG$.

    The code space associated to this quintuple is given by the direct sum
    \[
        \bigoplus_{t \in T_0} \pi(t)\calC,
    \]
    where $\calC \simeq A \otimes B$. 
    The triple $(\mathcal L, \mathcal G, T_0)$ are referred to, respectively, as the bare logical group, the gauge group, and the classical logical operators. 
    We further define the {\it Clifford code operator algebra} to be the algebra 
    \[
        \mathcal A (G,\pi,\calL,\calG,T_0) = \bigoplus_{t \in T_0} \pi(t) ( \calB(A) \otimes I_{B}  ) \pi(t)^{-1} .
    \]
\end{definition} 

\subsection{Examples}

Here we revisit the projective error models given in Section~\ref{section:examples}, and we give explicit examples of Clifford codes in each case, both regular (non-hybrid) and hybrid versions, and stabilizer and non-stabilizer codes. 
For the sake of brevity, we will mainly focus on subspace codes (i.e., those without subsystem structures), but we note that each of the examples can be easily converted into subsystem code examples by using a standard construction from the theory of subsystem codes: converting a subset of logical qubits into gauge qubits. 

\begin{example}
    The motivating class of stabilizer codes, sometimes called the `canonical' stabilizer codes, arise in the Pauli error model through certain Abelian subgroups. The elements of $P_n$ either commute or anti-commute, and any Abelian subgroup $S$ that does include $-I$ with $s$ independent generators can be seen to have a $2^{k}$-dimensional stabilizer subspace where $k=n-s$. The canonical example of such a subgroup is $S = \langle Z_1, \ldots , Z_{s} \rangle$, where $Z_1 = Z \otimes I^{\otimes (n-1)}$, $Z_2 = I \otimes Z \otimes I^{\otimes (n-2)}$, etc; in other words, one gets $Z_j$ by setting $b_j = 1$, $a_j = 0$, and $a_i = 0 = b_i$ for all $i \neq j$, in the projective representation given in Eq.~\eqref{canonical}.

    This is a core example in the foundational stabilizer formalism of Gottesman for quantum error correction \cite{gottesman1996class,gottesman1997stabilizer}. Subsequently an extension of the formalism for subsystem codes was obtained by Poulin \cite{Poulin2005Stabilizer}, built on the `operator quantum error correction' framework \cite{Kribs2005Unified,Kribs2006oqec}. More recently, in \cite{DauphinaiKribsVasmer2024} both of these formalisms were extended to hybrid codes via the operator algebra approach discussed further below, and to the setting of entanglement-assisted codes in \cite{NadkarniAdonsouDauphinaisKribsVasmer2024}. Let us simply note here that a choice of coset transversal for the canonical codes is given by the $2^s$-element set: 
    \begin{align*} 
        \pi(T)
        & = \big\{ \pi ( a_1, 0, a_2, 0, \ldots , a_s, 0, 0, \ldots , 0) \,\, | \,\, a_1 , \ldots a_s \in \{ 0,1 \}   \big\} \\ 
        & = \big\{ X^{a_1} \otimes \ldots \otimes X^{a_s}\otimes I_2^{\otimes k} \,\, | \,\, a_1 , \ldots a_s \in \{ 0,1 \}   \big\}.
    \end{align*} 
    In stabilizer formalism theory, the elements of $T$ restricted to the code space map out all possible (orthogonal) syndrome spaces for the code. As shown in \cite{DauphinaiKribsVasmer2024,NadkarniAdonsouDauphinaisKribsVasmer2024} and discussed above, subsets of classical logical operators $T_0 \subseteq T$ generate hybrid codes. 
\end{example}

The following is built on the error model of Example~\ref{cdexample}. 

\begin{example}\label{nonstabilizer}
    Let $\{ \ket{00},\ket{01}, \ket{10}, \ket{11}\}$ be the standard basis of $\bbC^2 \otimes \bbC^2$. Then $\calC = \Span\{\ket{00}, \ket{01}\}$ is a Clifford subspace code for $G = C_2 \times D_{2d}$ and $\pi$ given in Eq.~(\ref{cdexampleeqn}), which furthermore is \emph{not} a stabilizer code. To see this, we compute the group of logical operators $L$, and the group of stabilizers $S$. Note that the element
    \begin{equation*}
        \pi ( c^k b^l a^m ) = X^kZ^m \otimes X^lP^m
    \end{equation*}
    maps $\calC$ to $\calC$ if and only if $k = 0$. From this we immediately have $L = D_{2d} = \langle a,b \rangle$. Furthermore, the only element of $L$ that acts as a scalar on $\calC$ is the identity element. Hence, $S = \{e\}$, the trivial subgroup of $G$. This already tells us that $\calC$ is not a stabilizer code, since a stabilizer code is necessarily the common eigenspace of its stabilizer group, cf. \cite[Proposition 6.7]{Eidesen2025}.

    To see that this is a Clifford subspace code, let $\gamma = \Res^G_L\pi|_\calC$, which we can view as a projective representation $\gamma \colon L \to \calU(\bbC^2)$, which is concretely given by 
    \begin{equation*}
        \gamma(b^la^m) = X^lP^m.
    \end{equation*}
    We recognize this as the XP-Projective error model. We again refer to the proof of \cite[Proposition 8.1]{Eidesen2025} to see that $\Ind_L^G\gamma \simeq \pi$, hence showing that $\calC$ is a Clifford subspace code.

    As this is a 2-dimensional code in 2-qubit space, there is limited opportunity for a hybrid encoding. In this case, $T = \{e,c\}$ is a coset transversal, and $\pi(T) = \{ I_2 \otimes I_2, X \otimes I_2 \}$. Note though, if we extend the projective error model by taking products, i.e. consider the projective error model $(G^n,\pi^{\otimes n})$ on $(\bbC^2 \otimes \bbC^2)^{\otimes n}$, this projective error model will allow for larger codes (hybrid and otherwise), with transversal operators given by copies of $X_1$ on each (2-qubit) subsystem. We will revisit this example in the next section. 
\end{example}

The following is built on the projective error model of Example~\ref{zlexample}. 

\begin{example}
    Let $\{\ket{ij} : i \in \{0,1\}, j \in \{0,1,...,n-1\}\}$ be the standard basis of $\bbC^2 \otimes \bbC^n$. Then
    \begin{equation*}
        \calC = \Span\{ \ket{0j} : j \in \{0,1,...,n-1\}\}
    \end{equation*}
    is a Clifford subspace code for $G = \mathbb Z_2 \times L$, with $L = (\bbZ_n \times \bbZ_n) \rtimes \bbZ_2$, and $\pi$ given in Eq.~(\ref{zlexampleeqn}), which is \emph{not} a stabilizer code. Similar to the previous example we have that this is best shown by computing the group of logical operators and the stabilizer group of $\calC$. Note that the element
    \[
        \pi(a,b,c,d) = X^a Z^d \otimes (X_n)^b (Z_n)^c C^d
    \]
    maps $\calC$ to $\calC$ if and only if $a=0$. The rest of the argument is very similar to the previous example, and one finds that the group of logical operators is precisely $L = (\bbZ_n \times \bbZ_n) \rtimes \bbZ_2$, and the group of stabilizers $S$ is again the trivial subgroup. Hence, $\calC$ is not a stabilizer code, which follows by \cite[Proposition 6.7]{Eidesen2025}. 

    Let $\gamma := \Res^G_L\pi|_\calC$, i.e. $\gamma \colon L \to \calU(\bbC^n)$ is the projective representation of $L$ given by
    \begin{equation*}
        \gamma(b,c,d) = (X_n)^b(Z_n)^cC^d.
    \end{equation*}
    One finds that $\Ind_L^G\gamma \simeq \pi$ (the proof of this is completely analogous to the proof of the corresponding fact in the previous example, the details for which can be found in \cite[Proposition 8.1]{Eidesen2025}). Thus, $\calC$ is a Clifford subspace code.

The code here is an $n$-dimensional code in a $2n$-dimensional space, and a coset transversal is given by $T = \{(0,0,0,0),(1,0,0,0)\}$ where $\pi(T) = \{ I_{2} \otimes I_{n}, X\otimes I_n \}$. This gives limited opportunity for hybrid encodings, but of course considering tensor products of the representation expands these possibilities on the larger systems.  
\end{example}

Even though the framework of hybrid Clifford subsystem codes we have given captures many existing examples in the literature, as well as new examples, there are related examples of codes that are not captured by this framework. The following is built on the XP error model of Example~\ref{example:XP-projective error model} and gives an illustration of this. Note that even though the following does not define a Clifford subspace code, it defines a \emph{weak} stabilizer code as defined in \cite[Definition 5.2]{Eidesen2025}.

\begin{example}\label{example: negative XP result}
    Let $\{\ket{x} : x \in \{0,1\}^n\}$ denote the standard basis for $\calH := (\bbC^2)^{\otimes n}$, $G := (D_3)^n$, and $\pi \colon G \to \calU(\calH)$ be given by
    \[
        \pi(b^{k_1}a^{l_1},\ldots,b^{k_n}a^{l_n}) = X^{k_1}P^{l_1} \otimes \cdots \otimes X^{k_n}P^{l_n},
    \]
    as in Example~\ref{example:XP-projective error model} with $d = 3$. Consider the subgroup
    \[
        S := \{ (b^ka^l,b^ka^l,\ldots,b^ka^l) \in G : k \in \{0,1\}, \ l \in \{0,1,2\}\},
    \]
    and the function $\lambda \colon S \to \bbT$ defined by $\lambda(b^ka^l,b^ka^l,\ldots,b^ka^l) = \zeta_3^{-nl}$. Define
    \[
        \calC := \{ \ket{\psi} \in \calH : \pi(s)\ket{\psi} = \lambda(s)\ket{\psi} \text{ for all } s \in S \}
    \]
    Since $S$ is not a normal subgroup of $G$ we have that this defines the \emph{weak} stabilizer code associated to the quadruple $(G,\pi,S,\lambda)$, cf.~\cite[Definition 5.2]{Eidesen2025}, if $\calC \neq 0$. Note that $\calC$ can also be described as the common eigenvalue $1$ eigenspace of the operators $X^{\otimes n}$ and $\zeta_3^n P^{\otimes n}$.

    One computes that $\Res^G_S\sigma_\pi = \delta \lambda$, which by \cite[Theorem 5.5]{Eidesen2025} is a necessary condition for $\calC \neq 0$. We can then appeal to \cite[Proposition 5.7]{Eidesen2025} to compute that
    \[
        \dim\calC
        = \frac{1}{|S|} \sum_{s \in S} \overline{\lambda(s)}\chi_\pi(s)
        = \frac{1}{6} \sum_{l = 0}^{2} \zeta_3^{nl}(1 + \zeta_3^l)^n
        = \frac{2^n + 2(\zeta_3 + \zeta_3^2)^n}{6} 
        = \frac{2^n + 2(-1)^n}{6}.
    \]
    A basis for $\calC$ is given by $\{ \ket{x} + \ket{\overline{x}} : \wt(x) \equiv 2n (\!\!\!\mod 3) \}$ where $\wt(x) = \sum_{i = 1}^{n} x_i$, and $\overline{x} \in \{0,1\}^n$ is defined by setting $\overline{x}_i = x_i + 1$ and reducing modulo $2$. In other words, $\ket{\overline{x}} = X^{\otimes n} \ket{x}$. For example, if $n = 2$, this basis is given by the single vector $\ket{01} + \ket{10}$, and if $n = 3$ its given by the single vector $\ket{000} + \ket{111}$.
    
    With this basis for $\calC$ in hand, we compute that the group of logical operators $L_{(G,\pi)}(\calC)$ is equal to $S$. If $\calC$ was a Clifford subspace code we would have, by \cite[Corollary 6.5]{Eidesen2025}, that
    \[
        \dim\calH \cdot |L_{(G,\pi)}(\calC)| = \dim\calC \cdot |G|.
    \]
    Computing these numbers we get that
    \[
        2^n \cdot 6 \neq \frac{2^n + 2(-1)^n}{6} \cdot 6^n,
    \]
    hence $\calC$ is \emph{not} a Clifford subspace code.
    
    Asymptotically, we have that $\calC$ encodes $n-\log_26$ qubits of information. This is better than the $n-3$ qubits encoded by optimal stabilizer codes, but not quite as good as the CWS codes introduced in \cite{Smolin2007} on which this code is based. However, those codes do not fit into a stabilizer framework.
    
    Following \cite{nemec2021infinite}, we can construct non-Clifford subsystem and hybrid codes by extending these codes by one qubit.
    By this, we mean a one-qubit gauge subsystem is appended to the end of the code, potentially followed by an entangling unitary between the code and the gauge subsystem.
    Doing this, we obtain a non-trivial Clifford subsystem code with a one-qubit gauge subsystem, with gauge group
    \begin{equation*}
        \mathcal{G}=\left\langle X^{\otimes n}\otimes I_2,\zeta_3^{n}P^{\otimes n}\otimes I_2,I_2^{\otimes n}\otimes X, I_2^{\otimes n-1}\otimes Z^{\otimes 2}\right\rangle.
    \end{equation*}
    This gauge qubit can be gauge fixed, giving us a hybrid Clifford subspace code encoding one classical bit, with classical logical operators 
    \begin{equation*}
        T_0=\left\{I_2^{\otimes n+1},I_2^{\otimes n-1}\otimes Z^{\otimes 2}\right\}.
    \end{equation*}
\end{example} 

\section{Error Correction Theorem}\label{section: error correction theorem}

In this section, we first derive an explicit form for the code space projection and consequently the projections defined by the different classical logical operators. Then we prove the main error correction theorem for these codes. 

\subsection{The Code Projection and Classical Logical Operator Orthogonality}

The following results are well-known. We present short proofs of these facts for the sake of completeness. The proofs presented here are modeled after the proofs in the unitary setting presented in \cite{Serre77}. Recall that the character $\chi_\pi$ of a $\sigma$-projective representation $\pi$ of $G$ is defined by
\begin{equation*}
    \chi_\pi(x) = \Tr(\pi(x)) \,\, \mathrm{for}\,\mathrm{all}\, x\in G.
\end{equation*}

\begin{definition}
We say that a funtion $f: G \rightarrow \mathbb C$ is a {\it $\sigma$-class function} if, for all $x,y\in G$,  
\begin{equation*}
    f(yxy^{-1}) = \overline{\sigma(y,x)}\sigma(yxy^{-1},y)  \, f(x).
\end{equation*}
\end{definition}

\begin{lemma} 
The character $\chi_\pi$ is a $\sigma$-class function. 
\end{lemma} 

\begin{proof} 
Let $x,y\in G$. Then we have 
\begin{eqnarray*}
    \chi_\pi(yxy^{-1}) &=& \Tr(\pi(yxy^{-1})) \\ 
    &=&  \overline{\sigma(yx,y^{-1})} \, \overline{\sigma(y,x)} \, \sigma(y,y^{-1}) \, \Tr(\pi(y)\pi(x) \pi(y^{-1})) \\
    &=&  \overline{\sigma(y,x)} \, \sigma(yxy^{-1},y) \,       \chi_\pi(x) , 
\end{eqnarray*}
where in this calculation we have used Eq.~(\ref{projectiverelation}), Eq.~(\ref{inverse}), the relation $\sigma(y,y^{-1}) = \sigma(y^{-1},y)$, and the cocycle equation 
\[
\sigma(yx,y^{-1}) \, \sigma(yxy^{-1},y) = \sigma(yx,e) \, \sigma(y^{-1},y) = \sigma(y, y^{-1}), 
\]
recalling that $ \sigma(yx,e) = 1$. 
\end{proof}

We further note that the character on $L$ given by $\chi_\gamma(x) = \Tr(\gamma(x))= \Tr(\pi(x)P_{\mathcal C})$ for $x\in L$ is also a $\sigma$-class function on $L$. This follows from a similar calculation, with the addition of $P_{\mathcal C}$ which commutes with all operators in $L$. 

\begin{lemma}\label{lemma:intertwiners from class functions}
    Let $G$ be a finite group and $\pi$ be a $\sigma$-projective representation of $G$ on $\calH$. If $f$ is a $\sigma$-class function, then the linear map $T_f \colon \calH \to \calH$ defined by
    \begin{equation*}
        T_f = \sum_{x \in G}\overline{f(x)}\pi(x),
    \end{equation*}
    is an intertwiner. Furthermore, if $\pi$ is irreducible,
    \begin{equation*}
        T_f = \frac{|G|}{\dim\calH}\langle \chi_\pi, f \rangle \id_\calH.
    \end{equation*}
\end{lemma}
\begin{proof}
    Let $y \in G$. Then we have, 
    \begin{align*}
        \pi(y)T_f
        & = \sum_{x \in G} \overline{f(x)} \pi(y)\pi(x) \\
        & = \sum_{x \in G} \overline{f(x)}\sigma(y,x)\pi(yx) \\
        & = \sum_{x \in G} \overline{f(x)}\sigma(y,x)\pi(yxy^{-1}y) \\
        & = \sum_{x \in G} \overline{\sigma(y,x)}\sigma(yxy^{-1},y)\overline{f(yxy^{-1})}\sigma(y,x)\overline{\sigma(yxy^{-1},y)}\pi(yxy^{-1})\pi(y) \\
        & = \sum_{x \in G} \overline{f(yxy^{-1})}\pi(yxy^{-1})\pi(y) \\
        & = T_f\pi(y).
    \end{align*}
    Hence, $T_f$ is an intertwiner. If $\pi$ is irreducible we have by Schur's Lemma that $T_f = \lambda\id_\calH$ for some $\lambda \in \bbC$. Computing the trace of both of these sides (and using Eq.~(\ref{eqn:inner product})) yields the desired equality.
\end{proof}

We next derive the formula for the code space projection. 

\begin{proposition}\label{lemma:code projection}
    Let $\calH$ be a finite dimensional Hilbert space and $(G,\pi)$ be a projective error model on $\calH$. Suppose that $\calC = \calC(G,\pi,L,\gamma)$ is a Clifford subspace code. Then the orthogonal projection onto $\calC$ is given by
    \begin{equation*}
        P_\calC = \frac{\dim\calC}{|L|}\sum_{x \in L}\overline{\chi_\gamma(x)}\pi(x).
    \end{equation*}
\end{proposition}
\begin{proof}
    To simplify notation let
    \begin{equation*}
        T = \frac{\dim\calC}{|L|}\sum_{x \in L}\overline{\chi_\gamma(x)}\pi(x)
    \end{equation*}
    and write $\calH = \calC \oplus \calC^\perp$. To show that $T$ is indeed equal to the projection onto the code, it is enough to show that $T\ket{\psi} = \ket{\psi}$ for any $\ket{\psi} \in \calC$ and that $T\ket{\varphi} = 0$ for any $\ket{\varphi} \in \calC^\perp$. To show that the first holds, let $\ket{\psi} \in \calC$. Then,
    \begin{align*}
        T \ket{\psi}
        & = \frac{\dim\calC}{|L|}\sum_{x \in L}\overline{\chi_\gamma(x)}\pi(x)\ket{\psi} & \\
        & = \frac{\dim\calC}{|L|}\sum_{x \in L}\overline{\chi_\gamma(x)}\gamma(x)\ket{\psi} & \text{since } \pi(x)|_\calC = \gamma(x), \text{ see \eqref{eqn:clifford code representation}}, \\
        & = \langle\chi_\gamma, \chi_\gamma\rangle\ket{\psi} & \text{by Lemma~\ref{lemma:intertwiners from class functions}}, \\
        & = \ket{\psi} & \text{by Schur's Lemma and \eqref{eqn:inner porduct of characters}}.
    \end{align*}

    Now let $\ket{\varphi} \in \calC^\perp$ and let $V$ denote the smallest $L$-invariant subspace of $\calH$ containing $\ket{\varphi}$. By Maschke's Theorem we may assume that $\Res^G_L\pi|_V$ is an irreducible projective representation of $L$ on $V$. Let $\delta := \Res^G_L\pi|_V$, then by the same computation as above we see that
    \begin{equation*}
        T \ket{\varphi} = \frac{\dim\calC}{\dim V}\langle\chi_\delta, \chi_\gamma\rangle\ket{\varphi}.
    \end{equation*}
    By Schur's Lemma we have that
    \begin{equation*}
        \langle\chi_\delta, \chi_\gamma\rangle =
        \begin{cases}
            1, & \delta \simeq \gamma, \\
            0, & \delta \not\simeq \gamma.
        \end{cases}
    \end{equation*}
    By construction we have that $\delta \neq \gamma$, but they might still be isomorphic. To see that they cannot be isomorphic, notice that both are subrepresentations of $\Res^G_L\pi$. Combining Maschke's Theorem and Schur's Lemma we have that the number of subrepresentations of $\Res^G_L\pi$ that are isomorphic to $\gamma$ equals
    \begin{equation*}
        \dim \Hom_L(\gamma, \Res^G_L\pi).
    \end{equation*}
    We compute that
    \begin{align*}
        \dim \Hom_L(\gamma, \Res^G_L\pi)
        & = \dim \Hom_G(\Ind_L^G\gamma, \pi) & \text{by Frobenius reciprocity}, \\
        & = \dim \Hom_G(\pi, \pi) & \text{since $\calC$ is a Clifford subspace code}, \\
        & = 1 & \text{by Schur's Lemma}.
    \end{align*}
    Hence, $\gamma$ is the only subrepresentation of $\Res^G_L\pi$ that is isomorphic to $\gamma$, from which we conclude that $T\ket{\varphi} = 0$.
\end{proof}

Given a coset transversal $T$ for $G / L$, we know from Proposition~\ref{prop:Hybrid structure of Clifford codes} that the subspaces $\pi(t)\mathcal C$, for $t\in T$, are mutually orthogonal and span $\mathcal H$. The projections onto these coset subspaces, subsets of which define the hybrid codes, have an explicit form as follows. 

\begin{corollary}\label{cor:cosetprojectionform}
    Let $\calH$ be a finite dimensional Hilbert space and $(G,\pi)$ be a projective error model on $\calH$. Suppose that $\calC = \calC(G,\pi,L,\gamma)$ is a Clifford subspace code and let $T$ be a coset transversal for $G/L$. Then the orthogonal projections $\{P_t: t\in T\}$ onto the pairwise orthogonal subspaces $\pi(t) \mathcal C$ are given by:  
    \begin{equation*}
        P_t 
        =  \pi(t) \, P_\calC \, \pi(t)^{-1} 
        = \frac{\dim\calC}{|L|}\sum_{x \in L}\big( \overline{\chi_\gamma(t^{-1} x t)} \, \sigma_\pi( t^{-1} x t, t^{-1} ) \, \sigma_\pi( t, t^{-1} x ) \big) \, \pi(x).
    \end{equation*}
    Moreover, each projection $P_t$ is defined independently of the choice of coset representatives; that is, $P_t = P_{ts}$ for all $s \in L$.  
\end{corollary}
\begin{proof} 
    The equation for $P_t$ follows from Proposition~\ref{lemma:code projection} and the definition of the function $\sigma_\pi$, together with the fact that $L$ is a normal subgroup. Note that the cocycle relation also gives us that:
    \begin{equation*}
        \sigma_\pi( t^{-1} x t, t^{-1} ) \, \sigma_\pi( t, t^{-1} x ) = \sigma_\pi( t, t^{-1} x t ) \, \sigma_\pi( x t, t^{-1} )
    \end{equation*}
    which gives us another presentation of $P_t$.

    For the last statement, note that for $s\in L$, we have  
    \[
        \pi(ts) \, P_\calC \, \pi(ts)^{-1} = (\overline{\sigma_\pi(t,s)} \pi(t) \pi(s) ) \, P_\calC \, ( \overline{\sigma_\pi(t,s)}\pi(t) \pi(s) )^{-1}  
        =   \pi(t) \, ( \pi(s) P_\calC \pi(s)^{-1} ) \, \pi(t)^{-1} .  
    \]
    For $s, x \in L$, we can use the fact that $\chi_\gamma$ is a $\sigma_\pi$-class function on $L$ (since $\sigma_\gamma = \Res^G_L\sigma_\pi$), to compute that
    \begin{align*}
        \overline{\chi_\gamma (x)} \,\pi(s) \pi(x) \pi(s)^{-1} 
        & = \big( \overline{\chi_\gamma (x)} \sigma(s,x) \sigma(sx,s^{-1}) \overline{\sigma(s,s^{-1})} \big) \,\pi(s x s^{-1}) \\ 
        & = \big( \overline{\chi_\gamma (sxs^{-1})} \overline{\sigma(s,x)} \sigma(sxs^{-1}, s)  \sigma(s,x) \sigma(sx,s^{-1}) \overline{\sigma(s,s^{-1})} \big) \,\pi(s x s^{-1}) \\ 
        & = \overline{\chi_\gamma (sxs^{-1})}  \,\pi(s x s^{-1})  .
    \end{align*}
    It follows from this calculation and the formula for $P_{\mathcal C}$ that $\pi(s) \, P_\calC \, \pi(s)^{-1} = P_\calC$, and so $P_t =  \pi(ts) \, P_\calC \, \pi(ts)^{-1}$ as claimed. 
\end{proof} 

\subsection{Error Correction Theorem}

As in the recent work \cite{DauphinaiKribsVasmer2024}, we shall formulate and prove the main error correction theorem for these Clifford codes from the operator algebra quantum error correction perspective \cite{Beny2007Generalization,Beny2007Quantum}. We first give a condensed review of the main error correction result in that framework.   
The starting point is a reformulation of quantum error correction from the Heisenberg picture for quantum dynamics (i.e., focus on the time evolution of observables), rather than the more traditional Schr\"odinger picture (i.e., focus on evolution of states) used in the subject. Mathematically, this naturally leads to the notion of codes identified with operator algebras $\mathcal A$, which allow, as discussed in Section~\ref{prelims}, an embracing of regular subspace and subsystem codes along with their hybrid classical and quantum varieties. 
See \cite{Beny2007Generalization,Beny2007Quantum,DauphinaiKribsVasmer2024} for more details on these points. 

For our purposes here, we simply recall the following description of correctability: An algebra $\mathcal A$, with unit projection $Q$, is {\it correctable} for a set of error operators $\mathcal E = \{ E_a \}$ if and only if for all $X \in \mathcal A,$ and for all $a,b$, we have the commutator relation: 
\[
    [Q E_a^* E_b Q, X] = 0 . 
\]
When translated back to the Schr\"odinger picture, in the case that $\mathcal A = P_{\mathcal C} \calB(\mathcal H) P_{\mathcal C}$ for some subspace $\mathcal C$ of $\mathcal H$, we recover standard (subspace) correctable codes, and then subsystem codes when the subspace has a tensor decomposition. Algebras with direct sums in their decompositions allow for hybrid features in the encodings. Explicitly, in the Schr\"odinger picture, an operator algebra $\mathcal A$ is correctable for an error map $\mathcal E$ if and only if there exists a recovery map $\mathcal R$ such that for any density operator $\rho = \sum_i \alpha_i (\tau_i \otimes \rho_i)$ with $\tau_i \in \mathcal T(A_i)$, $\rho_i \in \mathcal T(B_i)$ (the trace class operators on $A_i,B_i$), and nonnegative scalars $\sum_i \alpha_i=1$, there are density operators $\tau_i' \in \mathcal T(A_i)$ for which
\begin{equation}\label{schrodingpicture}
    (\mathcal R\circ\mathcal E)(\rho) 
    = \sum_i \alpha_i \mathcal R \bigl({\mathcal E \bigl({\rho_i \otimes \tau_i }\bigr)}\bigr) 
    = \sum_i \alpha_i (\rho_i \otimes \tau_i' ).
\end{equation}
It is this notion of correctability that we focus on below. 

\begin{definition}
    Let $\calH$ be a finite dimensional Hilbert space and $(G,\pi)$ be a projective error model on $\calH$. Suppose that $A \otimes B = \calC(G,\pi,\calL,\calG,T_0)$ is a hybrid Clifford subsystem code. We say that the code is \emph{correctable} for a set of errors $\{g_a\}_{a} \subset G$ if the Clifford code operator algebra $\mathcal A (G,\pi,\calL,\calG,T_0)$ is correctable for the set $\{\pi(g_a)\}_{a}$.
\end{definition} 

It follows from the result quoted above and the structure of operator algebras and their commutants discussed in Section~\ref{prelims}, that an algebra $\mathcal A$ with decomposition as given in  Section~\ref{prelims}  is correctable for $\{ E_a \}$ if and only if for all $a,b$ there are operators $X_{abi} \in \calB(B_i)$ such that
\begin{equation}\label{oaqeceqn}
Q E_a^* E_b Q = \sum_i I_{A_i} \otimes X_{abi},
\end{equation}
where here the operators $I_{A_i} \otimes X_{abi}$ are understood to act on $A_i \otimes B_i$, and so the sum (when there is a sum with more than one term) is thus an orthogonal direct sum of operators. 

The case of a sum with a single term captures the Knill-Laflamme error correction conditions \cite{knill1997theory} when $\dim B_1 = 1$ (and so $X_{kl1}$ are complex scalars), and the operator quantum error correction testable conditions \cite{Kribs2005Unified,Kribs2006oqec} when $\dim B_1 > 1$. We note the special case of the right side of Eq.~(\ref{oaqeceqn}) obtained when there is a bona fide sum and there are no subsystem structures within the summand spaces. If $\{P_i \}$ are the projections onto the different subspaces, then $Q= \sum_i P_i$, and the conditions become: 
\begin{equation*}
    P_i E_a^* E_b P_j = \delta_{ij} \lambda_i P_i,
\end{equation*}
for some scalars $\lambda_i$, which one can see obtain by multiplying the left and right of both sides of Eq.~(\ref{oaqeceqn}) by $P_i$ and $P_j$ respectively. (In the case that the $P_i$ have a common rank, these conditions were obtained independently in \cite{grassl2017codes} as noted in \cite{majidy2018unification}.) 

Recall that for a Clifford subsystem code $\calC$, we assume that the group of stabilizers $S_{(G,\pi)}(\calC)$ is a subgroup of the gauge group $\calG$.

\begin{theorem}\label{errorcorrecthm}
    Let $\calH$ be a finite dimensional Hilbert space and $(G,\pi)$ be a projective error model on $\calH$. Suppose that $\calC = A \otimes B = \calC(G,\pi,\calL,\calG,T_0)$ is a hybrid Clifford subsystem code where the hybrid structure is given by the classical logical operators $\{t_i\}_{i} = T_0$. Let $L = \calL\calG$ denote the group defining the Clifford subspace structure of $\calC$. Then a set of errors $\{g_a\}_{a} \subset G$ is correctable if and only if, for all $a$ and $b$,
    \begin{equation}\label{eqn:detectability condition}
        g_a^{-1}g_b \notin \big(L \setminus \calG\big) \bigcup \left( \bigcup_{i \neq j} t_i^{-1} \, L \, t_j \right).
    \end{equation}
\end{theorem}
\begin{proof}
    Note that each transversal subspace has a subsystem structure $A_i \otimes B_i := \pi(t_i) \mathcal C$ induced by the subsystem structure of $\mathcal C = A \otimes B$ and the unitary action of $\pi(t_i)$. Recall further that $P_i = \pi(t_i) P_{\mathcal C} \,\pi(t_i)^{-1}$ is the projection onto $\pi(t_i) \mathcal C$ by Corollary \ref{cor:cosetprojectionform}. In particular, this gives an identification of operators in $\pi(t_i) (\calB(A)\otimes \calB(B)) \pi(t_i)^{-1}$ with operators in $\calB(A_i)\otimes \calB(B_i)$. Specifically note that for any $X_B \in \calB(B)$ this maps $I_A \otimes X_B$ to $I_{A_i}\otimes X_{B_i}$ for some $X_{B_i}\in \calB(B_i)$.
    
    Let $Q = \sum_i P_i$, this is the unit projection in $\mathcal A (G,\pi,\calL,\calG,T_0)$. Given a $g \in G$, we want to show that there exists $X_{gi} \in \calB(B_i)$ for each $i$ such that
    \[
        Q\pi(g)Q = \sum_i I_{A_i} \otimes X_{gi}
    \]
    if and only if \eqref{eqn:detectability condition} holds for the element $g$ in place of $g_ag_b^{-1}$. By Proposition \ref{prop:Hybrid structure of Clifford codes} it suffices to show that
    \[
        P_i \pi(g) P_j =
        \begin{cases}
            I_{A_i} \otimes X_{gi} & \text{if } i = j, \\
            0 & \text{if } i \neq j.
        \end{cases}
    \]
    if and only if \eqref{eqn:detectability condition} holds for the element $g$. We first focus on the case where $i = j$. Then we have that
    \[
        P_i \pi(g) P_i
        = \pi(t_i) \big( P_\calC \pi(t_i)^{-1} \pi(g) \pi(t_i) P_\calC \big) \pi(t_i)^{-1}.
    \]
    By the proof of \cite[Theorem 6.3]{Eidesen2025} we have that $P_\calC \pi(t_i)^{-1} \pi(g) \pi(t_i) P_\calC = 0$ if and only if $t_igt_i^{-1} \not\in L$, since $L$ is a normal subgroup of $G$ this is equivalent to saying that $g \not\in L$. If $t_igt_i^{-1} \in L$ then we have, up to a scalar multiple of the identity, that
    \[
        P_\calC \pi(t_i)^{-1} \pi(g) \pi(t_i) P_\calC
        = P_\calC \pi(t_i^{-1}gt_i) P_\calC
        = (\alpha \otimes \beta)(t_i^{-1}gt_i),
    \]
    where the last equality follows since $\calC = A \otimes B$ is a Clifford subsystem code. It is then clear that 
    \[
        P_\calC \pi(t_i)^{-1} \pi(g) \pi(t_i) P_\calC = I_A \otimes X_g
    \]
    for some $X_g \in \calB(B)$ if and only if $t_igt_i^{-1} \in \calG$. By the comments above, we then have that
    \[
        \pi(t_i) \big( P_\calC \pi(t_i)^{-1} \pi(g) \pi(t_i) P_\calC \big) \pi(t_i)^{-1}
        = \pi(t_i) \big( I_A \otimes X_g \big) \pi(t_i)^{-1}
        = I_{A_i} \otimes X_{gi}
    \]
    for som $X_{gi} \in \calB(B_i)$ if and only if $t_igt_i^{-1} \in \calG$. Since $\calG$ is a normal subgroup $L$, $\calG$ is in particular a normal subgroup of $G$. Hence, $t_igt_i^{-1} \in \calG$ if and only if $g \in \calG$.

    So far we have showed that
    \[
        P_i \pi(g) P_i = I_{A_i} \otimes X_{gi}
    \]
    for some $X_{gi} \in \calB(B_i)$ if and only if $g \not\in L \setminus \calG$. We will now focus on the case for $P_i \pi(g) P_j$ with $i \neq j$. We compute that
    \[
        P_i \pi(g) P_j
        = \pi(t_i) \big( P_\calC \pi(t_i)^{-1} \pi(g) \pi(t_j) P_\calC \big) \pi(t_j)^{-1}.
    \]
    As $\pi(t_i)$ and $\pi(t_j)^{-1}$ are unitary operators we have that this is zero if and only if $P_\calC \pi(t_i)^{-1} \pi(g) \pi(t_j) P_\calC = 0$. Again by the proof of \cite[Theorem 6.3]{Eidesen2025} we have that $P_\calC \pi(t_i)^{-1} \pi(g) \pi(t_j) P_\calC = 0$ if and only if $t_i g t_j^{-1} \not\in L$, which is the case if and only if $g \not\in t_i^{-1} L t_j$, which completes the proof.
\end{proof}

\begin{remark}
    In terms of extensions of previous results, most recently this result builds upon \cite[Theorem 6.3]{Eidesen2025}. Furthermore, note that \cite[Proposition 3.5]{Eidesen2025} gives a way of translating between error correction theorems for projective error models, and (linear) error models. This allows us to view Theorem \ref{errorcorrecthm} as a generalization of several results in the literature. Most notably, if we use the projective error model from Example~\ref{example:Pauli projective error model} we recover the usual Pauli error correction theorems; starting with Gottesman's original result \cite{gottesman1996class,gottesman1997stabilizer}, then Poulin's \cite{Poulin2005Stabilizer} for subsystem stabilizer codes, and the recent extension \cite{DauphinaiKribsVasmer2024} to hybrid and operator algebra codes. More generally, if there is no hybrid or subsystem structure present, we recover the error correction theorems found in \cite{KnillI96,KnillII96,KlappeneckerRottelerII02} for regular Clifford codes. If there is a subsystem structure we furthermore recover the error correction theorem from \cite{klappenecker2010clifford}. 

        Let us also note the complementary set version of Eq.~\eqref{eqn:detectability condition}, which says that 
    \begin{equation}\label{eqn:detectability condition alternate}
        g_ag_b^{-1} \in \big((G \setminus L) \cup \calG \big) \cap \left(\bigcap_{i \neq j} G \setminus (t_i^{-1}Lt_j)\right), 
    \end{equation}
    and recall that we have the set equalities $t_i^{-1}L\, t_j = t_i^{-1} t_j \, L = L \, t_i^{-1}t_j$ as $L$ is normal. 
\end{remark} 

We discuss an illustrative example coming from one of our non-stabilizer codes presented above.

\begin{example}\label{example:correction for XPZ}
    Recall the projective error model $(G,\pi)$ on $\bbC^2 \otimes \bbC^2$ defined in Example \ref{cdexample}, where $G = C_2 \times D_{2d}$ and
    \[
        \pi(c^kb^la^m) = X^kZ^m \otimes X^lP^m.
    \]
    Here $\langle c \rangle = C_2$ and $\langle a,b \rangle = D_{2d}$.
    
    Consider the $n$-fold extension of the code in Example~\ref{nonstabilizer}; that is, we are considering the projective error model $(G^n,\pi^{\otimes n})$ on $\mathcal H = (\mathbb C^2\otimes \mathbb C^2)^{\otimes n}$ with the $2^n$-dimensional code $\overline{\mathcal C} = \mathcal C^{\otimes n}$. Recall here that $\mathcal C = \mathrm{span}\{ \ket{00}, \ket{01} \}$ and $\calC = \calC(G,\pi,L,\gamma)$ with $L = D_{2d}$ and
    \[
        \gamma(b^la^m) = X^lP^m.
    \]
    Then $\overline{\calC} = \calC(G^n,\pi^{\otimes n},L^n,\gamma^{\otimes n})$, which is a Clifford subsystem code by \cite[Proposition 10.1]{Eidesen2025}. For ease of notation we will relabel $G^n$ by $G$, $\pi^{\otimes n}$ by $\pi$, $L^n$ by $L$, and $\gamma^{\otimes n}$ by $\gamma$.
    
    A coset transversal for $G / L$ consists of $2^n$ elements, the most canonical choice of such a coset transversal is
    \[
        T = \{ (c_j^{k_j})_{j = 1}^{n} : k_j \in \{0,1\}\}.
    \]
    Then,
    \[
        \pi(T) = \{ (X^{k_1} \otimes I_2) \otimes \cdots \otimes (X^{k_n} \otimes I_2) : k_j \in \{0,1\}\}.
    \]

    For a specific example of a code and application of the theorem above, let $n=2$, so 
    \[
        \overline{\calC} = \mathrm{span}\{ \ket{0000}, \ket{0001}, \ket{0100}, \ket{0101}  \}.
    \]
    If we set $\calL = \{e\} \times D_{2d} \times \{e\} \times \{e\}$, $\calG = \{e\} \times \{e\} \times \{e\} \times D_{2d}$, then we have that $\overline{\calC} \simeq A \otimes B$ with 
    \[
        A = \Span\{\ket{00}, \ket{01}\} \quad \mathrm{and} \quad B = \Span\{\ket{00}, \ket{01}\}, 
    \]
    defining a subsystem structure on $\overline{\calC}$, i.e. $A \otimes B = \calC(G,\pi,\calL,\calG)$. If we really think of $\calH$ as four qubits, then $\overline{\calC}$ is obtained by restricting to the second and fourth qubit, where the second qubit corresponds to the logical qubit, and the fourth corresponds to the gauge qubit. Note that
    \begin{align*}
        \pi(\calL) & = \{ (Z^m \otimes X^lP^m)\otimes(I_2 \otimes I_2) : l \in \{0,1\}, \ m \in \{0,1,\ldots,2d\} \}, \text{ and} \\
        \pi(\calG) & = \{ (I_2 \otimes I_2)\otimes(Z^m \otimes X^lP^m) : l \in \{0,1\}, \ m \in \{0,1,\ldots,2d\} \},
    \end{align*}
    which makes it more clear how the logical group and the gauge group acts on the Clifford subsystem code. We may further choose a classical logical operator subset $T_0$ of the coset transversal $T$ such that
    \[
        \pi(T_0) = \{(I_2 \otimes I_2 \otimes I_2 \otimes I_2), (X \otimes I_2 \otimes I_2 \otimes I_2)\},
    \]
    which defines a hybrid Clifford subsystem code $\calC(G,\pi,\calL,\calG,T_0)$. The associated hybrid code space is  $\overline{\calC} \oplus X_{1,1} \overline{\calC} \cong (A\otimes B) \oplus X_{1,1}(A\otimes B)$, where $X_{1,1}$ denotes the Pauli X operator acting on the first qubit of the first (2-qubit) subsystem of $\mathcal H$, which is explicitly given by  
    \[
        \bigg(\Span\{\ket{0000}, \ket{0001}, \ket{0100}, \ket{0101}\}\bigg) \bigoplus \bigg(\Span\{\ket{1000}, \ket{1001}, \ket{1100}, \ket{1101}\}\bigg).
    \]

Now, to illustrate how error correction works for this code, consider the simple noise model for which the error $E = X_{2,1} = I_2 \otimes I_2 \otimes X \otimes I_2$ occurs with probability $0<p<1$, and no error occurs with probability $1-p$. The associated error map is given by $\mathcal E(\rho) = (1-p)\rho + p E\rho E^*$. 
Let $P_0$, respectively $P_1$, be the projection of $\mathcal H$ onto $\overline{\calC}$, respectively $X_{1,1} \overline{\calC}$. So the unit projection for the code algebra is $Q = P_0 + P_1$. 

First consider the hybrid subspace code case (i.e., we ignore the subsystem structure given by $A,B$ and consider the whole subspace for correction). Observe that the operator $E = E^* I$ belongs to the set given in Eq.~(\ref{eqn:detectability condition alternate}) for this code, and hence by Theorem~\ref{errorcorrecthm} the error map $\mathcal E$ is correctable. 
To see explicitly how recovery works here, note first that the subspaces $E \overline{\calC}$ and $E\, X_{1,1} \overline{\calC}$ are orthogonal to each other and to the full hybrid code subspace. This implies that the projection $Q_E = E Q E^*$ and $Q$ have orthogonal ranges, i.e., $Q_E Q = 0 = Q Q_E$.  
The subspaces $Q\mathcal H$ and $Q_E\mathcal H$ essentially act as `syndrome spaces' for the hybrid code in this setting, in the sense that they define a measurement that the recovery operation can be built upon. Indeed, if we put $R_0 = I^* Q = Q$ and $R_1 = E^* Q_E$, and define $\mathcal R(\rho) = R_0 \rho R_0^* + R_1 \rho R_1^*$, then for $\rho_i = P_i \rho_i P_i$, $i=0,1$, probabilities $0\leq p_0,p_1 \leq 1$, and $\rho = p_0 \rho_0 + p_1 \rho_1$ we have 
\begin{equation*}
(\mathcal R \circ \mathcal E) (\rho) = (1-p) \sum_{i,l} p_l R_i P_l \rho_l P_l R_i^*  + p \sum_{i,l} p_l R_i E P_l \rho_l P_l E^* R_i^*  
= \rho,
\end{equation*}
which makes use of the following operator relation versions of the subspace orthogonality relations noted above: For $l = 0,1$, 
\[
R_0 P_l = Q P_l = P_l, \quad R_0 E P_l = Q (I-Q) E P_l = 0 , \quad R_1 P_l = E^* Q_E Q P_l = 0, 
\]
and,
\[
R_1 E P_l = E^* Q_E E P_l = E^* (E Q E^*) E P_l = Q P_l = P_l. 
\]

If we reintroduce the $A,B$ subsystem structure to the code, we can only encode into $A$, but we can recover without worrying about noise on $B$. For instance, modifying the noise model above, we could consider instead $E = X_{2,1} Z_{2,1} P_{2,2} = I_2 \otimes I_2 \otimes (XZ) \otimes P$.
The full hybrid subspace code is no longer correctable for these errors, as $E$ acts non-trivially on the code space, but the 
hybrid subsystem code is correctable. Indeed, note that $Z_{2,1}$ acts noiselessly on the code states, whereas $P_{2,2}$ only acts non-trivially on subsystem $B$ and the copy of $B$ in the syndrome space for the code. Thus, it follows from the discussion above that the map $\mathcal R$ will act as a recovery for the errors $\mathcal E$ on the hybrid subsystem code, as in the general hybrid subsystem error correction recovery picture of Eq.~(\ref{schrodingpicture}).  
%    Applying Theorem \ref{errorcorrecthm}, one can verify that $\calC(G,\pi,\calL,\calG,T_0)$ is correctable for a subset of errors $\{g_a\}_a \subset G$ if and only if for all $a$ and $b$: 
%    \[
%        \pi(g_ag_b^{-1}) = 
%        \begin{cases}
%            (I_2 \otimes I_2) \otimes (X^{k}Z^m \otimes X^lP^m) & \text{or}, \\
%            (X \otimes I_2) \otimes (X^{k} \otimes I_2), &
%        \end{cases}
%    \]
%   where $k, l \in \{0,1\}$, and $m \in \{0,1,\ldots,2d\}$.
\end{example}

\begin{remark}
It is important to note that the recovery operation exhibited in the above example, and more generally built into Theorem~\ref{errorcorrecthm}, is a typical construction of an error correction procedure for these Clifford codes, and that it is very much in line with traditional quantum error correction recovery procedures. Indeed, if a hybrid code is correctable for a given set of error operators, then, up to a linear change of variables on the errors as in usual quantum error correction, the block of mutually orthogonal subspaces that define the hybrid code will be mapped (with probabilities) by the error operators to a set of blocks of mutually orthogonal subspaces that are themselves mutually orthogonal. (These will correspond to a subset of the subspaces given in the decomposition of Proposition~\ref{prop:Hybrid structure of Clifford codes}.)   A recovery operation can then be constructed as in the traditional case, by a measurement defined by the orthogonal blocks of subspaces, followed by a unitary reversal conditioned on the measurement outcome. 

Thus, even though the mathematics we needed to employ to formulate and investigate these Clifford codes has some sophisticated aspects to it, at a basic level, the actual error correction recovery works in the same way as it does for traditional recovery procedures. 
\end{remark}

Lastly, let us note that we can define a notion of code distance based on Theorem~\ref{errorcorrecthm}, and in particular from Eq.~(\ref{eqn:detectability condition}), when the underlying Hilbert space has a tensor subsystem structure. 

\begin{definition}\label{def:distance}
    Let $\calH$ be a finite-dimensional Hilbert space presented as the tensor product of subsystem Hilbert spaces, and let $(G,\pi)$ be a projective error model on $\calH$ for which the operators in $\pi(G)$ are represented as tensor products in the subsystem decomposition. Suppose that $\calC = A \otimes B = \calC(G,\pi,\calL,\calG,T_0)$ is a hybrid Clifford subsystem code with $T_0 = \{t_i\}_{i}$. Let $L = \calL\calG$ denote the group defining the Clifford subspace structure of $\calC$. We define the {\it code distance} $d(G,\pi,\calL,\calG,T_0)$ as follows:
    \begin{equation}\label{codedistance}
        d(G,\pi,\calL,\calG,T_0) = \min\mathrm{wt} \left( \pi\left(\big(L \setminus \calG\big) \bigcup \left( \bigcup_{i \neq j} t_i^{-1} \, L \, t_j \right)\right)\right), 
    \end{equation} 
    where $\min \mbox{wt}(A)$ is the minimum of the weights of operators in a set $A \subseteq \pi(G)$, and $\mbox{wt}(A)$ is the number of subsystems of $\mathcal H$ on which $A$ does not act as the identity operator.
\end{definition}

\begin{remark}
This notion of code distance includes a number of special cases of interest, some of which have been investigated. For non-hybrid codes, the union in Eq.~(\ref{codedistance}) is empty and one captures regular Clifford codes, for which the distance does not appear to have been significantly investigated at that general level. The distance for the subclass defined by the original stabilizer codes, and their subsystem and hybrid generalizations, has been extensively studied as it is central to many applications. A comprehensive analysis of the general code distance above is outside the scope of this work, and so we will leave such a potential examination for elsewhere. We note that the examples we have presented have a code distance simply equal to 1, and so more intricate examples would be required to truly investigate this distance and its properties. 
\end{remark}

\section{Conclusion}

In this work, we have added to the library of quantum error correcting codes, by extending the theory of Clifford codes to the setting of hybrid subspace and subsystem codes. Our generalization was two-fold, in that we have used it to inject projective representation theory more directly into the subject, and we have built it upon the operator algebra quantum error correction framework.

There are a number of potential lines of investigation that arise from this work. 
First note that the introduction of projective representation theory in quantum error correction leads to non-trivial phase factors between error operators to be more formally built into the framework. This allows us to see that for systems whose natural errors are not Pauli-type (e.g., certain continuous-variable systems, photon systems, or superconducting systems with specific symmetry breaking), codes can be directly constructed based on their physical symmetries, rather than forcibly twisting the physical errors into the Pauli framework. Additionally, one can consider operator algebra generalizations of these codes, which could allow for more connections with the rich theory of infinite-dimensional representation theory in operator algebras.

Within the class of codes we have introduced themselves, one can ask if there are extensions of other results from Clifford code theory, such as a representation theoretic description of which ones are stabilizer codes. We have included subsystem codes in our formulation, and we noted how one can construct them in our examples, but we have not investigated their properties deeply.  
We introduced a notion of code distance for these codes, but, as in, it seems, other Clifford code settings, we have not investigated it in detail. This is a direction to pursue for specific code subclasses of interest in applications, with the recent interest in subsystem codes as an example.  
One could also ask if there are entanglement-assisted versions of these codes \cite{NadkarniAdonsouDauphinaisKribsVasmer2024,defranco2026entanglement}. 
The operator algebra framework allows for a more seamless mix of the various hybrid and non-hybrid code types, and, as suggested above, it invites consideration of infinite-dimensional generalizations of Clifford codes. 
We leave these and related investigations for elsewhere.

\strut

{\noindent}{\it Acknowledgements.}
We thank the reviewer for helpful comments. 
J.E. acknowledges funding from The Research Council of Norway [Project 345433] and the Norwegian National Security Authority (NSM). An earlier draft of this paper will also appear as part of J.E.'s PhD-thesis completed at the University of Oslo under the supervision of Tron Omland, Erik B\'{e}dos, and Nadia S. Larsen, to whom J.E. is grateful for useful feedback and guidance.
D.W.K. was supported by Canada NSERC Discovery Grant RGPIN-2024-400160.

\bibliographystyle{plainurl}

\bibliography{refs}

\end{document}